\documentclass[11pt]{article}

\usepackage[all]{xy}           
\usepackage{amsmath,amsfonts,amssymb,amsthm,textcomp}
\usepackage{eucal}
\usepackage{epsfig}

\def\v #1{\vert #1\vert}             

\def\m #1 #2{(-1)^{{\v #1} {\v #2}}} 

\theoremstyle{plain}
\newtheorem{theorem}{Theorem}

\newtheorem{proposition}[theorem]{Proposition}

\newtheorem{example}{Example}

\theoremstyle{definition}
\newtheorem{definition}[theorem]{Definition}

\def\<#1>{\langle#1\rangle}

\usepackage{epsfig}
\usepackage{amsmath}
\usepackage{epic}
\usepackage{amssymb}
\usepackage{fancybox}
\usepackage{wrapfig}
\usepackage[font=footnotesize,labelsep=period]{caption}

\usepackage{float}

\usepackage{booktabs}

\usepackage[all]{xy}

\usepackage{fancyhdr}

\usepackage{fancyhdr}

\usepackage{appendix}

\usepackage{tikz-cd}

\begin{document}

\centerline{\Large \bf Geometry of the discrete Hamilton--Jacobi equation}\vskip 0.25cm
\centerline{\Large \bf  Applications in optimal control}

\medskip
\medskip

\centerline{M. de Le\'on and C.
Sard\'on}
\vskip 0.5cm
\centerline{Instituto de Ciencias Matem\'aticas, Campus Cantoblanco}\vskip 0.2cm
\centerline{Consejo Superior de Investigaciones Cient\'ificas}
\vskip 0.2cm
\centerline{C/ Nicol\'as Cabrera, 13--15, 28049, Madrid. SPAIN}

\medskip

\medskip

\begin{abstract}
   In this paper, we review the discrete Hamilton--Jacobi theory from a geometric point of view. 
In the discrete realm, the usual geometric interpretation of the Hamilton--Jacobi theory in terms
of vector fields is not straightforward. 

Here, we propose two alternative interpretations: one is the 
interpretation in terms of projective flows, the second is the temptative of constructing a discrete Hamiltonian vector field
renacting the usual continous interpretation.


Both interpretations are proven to be equivalent and applied in optimal control theory. 
The solutions achieved through both approaches are sorted out and compared by numerical computation. 

\end{abstract}

\section{Introduction}

The discretization of differential equations is efficient on frameworks
in which we cannot compute analytical solutions of the equation and numerical methods worked upon discretizations provide approximate solutions of our differential problem.

In recent years, there has been a growing effort to set proper discrete analogues of continuous models 
and design numerical methods to solve them.  
In this paper, we are interested in dynamical systems and optimal control problems endowed with a discrete Hamiltonian system.
Hence, numerical methods in geometric mechanics must preserve symplecticity since we work on a phase space, among some other restrictions.

The first inklings of discrete mechanics appeared in the realm of Lagrangian mechanics \cite{MarsWest}. 
The lack of a corresponding Hamiltonian theory lead to the development of discrete Hamiltonian mechanics. 
Since then, some works appeared on the discretization of Lagrangian and Hamiltonian systems on tangent and cotangent bundles, what lead
to variational principles for dynamical systems and principles of critical action on both the tangent and cotangent bundle \cite{GuiBloch,Marsden}. 
This gave rise to analogies between discrete and continuous symplectic forms, Legendre transformations, momentum maps and Noether's theorem.
The Hamiltonian side specially gave rise to optimal control problems by developing a discrete maximum principle that yields discrete necessary conditions for optimality.
Furthermore, discrete Hamiltonian theories have been particulary useful in distributed network optimization
and derivation of variational integrators \cite{Lall}. These constructions rely on numerical methods that do not only preserve symplecticity
but also the momentum map in the presence of symmetries. This is why the design of working numerical integrators is in vogue, since they do not necessarily
preserve conservation laws. 
The geometry of the space is also keypoint to perform better discretizations. For this matter, it is important to rely on symmetries and invariants of the geometric space.
For example, we examine conservation of energy, conservation of angular momentum, etc., when there exists a physical interpretation of the system under study.

In this work, we consider important to observe how objects differ from their continuous version if we implement a discretization of the system and how solutions are achieved by
minimazing the error in approximation. 

This is why we propose two different approaches for the same problem of obtaining a discrete, geometric Hamilton--Jacobi
theory. The passing from a continous to a discrete Hamilton--Jacobi theory is not straightforward as it might seem. Discrete
vector fields are new keypoint objects that need to be defined. 
Then, our outlook is twofold: on one hand, we propose a discrete geometric Hamilton--Jacobi theory 
interpreted in terms of discrete flows. This viewpoint
has not been devised in the literature before. On the other hand, we define a discrete Hamiltonian vector field and 
propose a Hamilton--Jacobi theory in terms of these discrete Hamiltonian vector fields.
Both approaches shall be used for the derivation of solutions of discrete Hamiltonians appearing in optimal control theories.

The goal is to reduce the amount of error derived from both approaches, to a level considered negligible for the modeling purposes at hand. 
Convergence between both approaches is numerically justified. In particular, it is shown how the second approach, or that of using 
a discrete vector field provides better approximations than the former discrete equation for the generating function.

So, the outline of the paper goes as follows: in Section 2, we review the common notation and fundamentals of classical continuous mechanics
and introduce paralell concepts on discrete mechanics briefly, alongside the continuous and discrete Hamilton--Jacobi equation.
Section 3 contains a discrete, geometric Hamilton--Jacobi theory that is twofold. First, we interpret
the Hamilton--Jacobi theory in terms of discrete flows, from
which we derive a discrete Hamilton--Jacobi equation. Second,
we propose an alternative discrete, geometric Hamilton--Jacobi theory in terms of a discrete Hamiltonian vector field.
Another discrete Hamilton--Jacobi equation is also derived. Both approaches are compared and proven equivalent.
Next, in Section 4 we propose a numerical example through a optimal control problem,  with which we show the convergence
between the two proposed methods and display the better outcome of the second proposal. 

To avoid mathematical conflict and without loss of generality, we assume all objects to be smooth and globally defined unless stated otherwise. Manifolds are connected and differentiable.

\section{Fundamentals}

\subsection{Continuous Mechanics}

\noindent
We consider the tangent bundle $TQ$ and the canonical projection $\tau_Q:TQ\rightarrow Q$.
A Lagrangian is a function $L:TQ\rightarrow \mathbb{R}$, where $L=L(q^i,\dot{q}^i)$ with
$(q^i)$ being coordinates on the manifold $Q$ and $(\dot{q}^{i})$ are the corresponding velocities.
We introduce the {\it Poincar\'e--Cartan 1-form} as
$$\theta_L=S^{*}(dL)=\frac{\partial L}{\partial {\dot q}^i}dq^i,$$
where $S=\frac{\partial}{\partial \dot{q}^i}\otimes dq^i$ is the canonical vertical endomorphism and $S^{*}$ denotes the adjoint operator. 
The {\it Poincar\'e--Cartan two-form} is defined as 
$$\omega_L=-d\theta_L$$
and the total energy of the system corresponds with $E_L=\Delta(L)-L\in C^{\infty}(TQ),$
where $\Delta=\dot{q}^i\frac{\partial}{\partial \dot{q}^i}$ is the {\it Liouville vector field} \cite{Crampin,Klein,LeonRodrigues}.
We say that $L(qî,\dot{q}^i)$ is {\it regular} if the Hessian matrix
\begin{equation}
 \left(W_{ij}\right)=\left(\frac{\partial^2 L}{\partial \dot{q}^i \partial \dot{q}^j}\right)
\end{equation}
is invertible. 
From here, we recover the classical expressions
$$\omega_L=dq^i\wedge dp_i,\quad \text{such that}\quad p_i=\frac{\partial L}{\partial \dot{q}^i},\quad E_L=\dot{q}^ip_i-L.$$

Geometrically, the Euler--Lagrange equations can be written in the symplectic way as
\begin{equation}
 \iota_{\xi_L}\omega_L=dE_L
\end{equation}
whose solution $\xi_L$ is called the {\it Euler--Lagrange vector field}. It is a second-order differential equation (SODE, for short); indeed if we write the
Euler--Lagrange vector field explicitly,
\begin{equation}
 \xi_L=q^{i}\frac{\partial}{\partial q^i}+\xi_i(q_i,\dot{q}_i)\frac{\partial}{\partial {\dot q}^i}
\end{equation}
its integral curves $(q^i(t),\dot{q}^i(t))$ are lifts of their projections $(q^i(t))$ on $Q$ and are solutions of the system of differential equations
\begin{align}
 \frac{dq^i(t)}{dt}&=\dot{q}^i,\nonumber\\
 \frac{d\dot{q}^i(t)}{dt}&=\xi^i,
\end{align}
which is equivalent to 
\begin{equation}
 \frac{d^2 q^i(t)}{dt^2}=\xi^i.
\end{equation}
The curves $q(t)$ in $Q$ are called the solutions of $\xi_L$ that correspond with the solutions of the Euler--Lagrange equation
\begin{equation}
 \frac{d}{dt}\left(\frac{\partial L}{\partial \dot{q}^i}\right)=\frac{\partial L}{\partial q^i},\quad 1\leqslant i\leqslant n=\text{dim} Q.
\end{equation}


The passing from the Lagrangian to the Hamiltonian setting is introduced by a Legendre transformation, as the fibered mapping $FL:TQ\rightarrow T^{*}Q$ 
such that $\pi_Q\circ FL=\tau_Q$. Here, $T^{*}Q$ is the cotangent bundle of $Q$ with canonical projection $\pi_Q:T^{*}Q\rightarrow Q$. A simple
computation shows that {\it FL} is a local diffeomorphism if and only if $L$ is regular. We say that the Lagrangian is {\it hyperregular} if the Legendre transform $FL(q^i,\dot{q}^i)=(q^i,p_i)$ is a global diffeomorphism.
From now on, and since this is the usual case in Mechanics, we will assume that $L$ is hyperregular.
The Hamiltonian is retrieved through $H(q^i,p_i)=E_L \circ FL^{-1}$.

If $\omega_Q$ is the canonical symplectic form on $T^{*}Q$ where $(q^i,p_i)$ are the canonical coordinates $T^{*}Q$, then $\omega_Q=dq^i\wedge dp_i$ and therefore
\begin{equation}
 \iota_{X_H}\omega_Q=dH
\end{equation}
is the geometric Hamilton equation, where the Hamiltonian vector field $X_H$ on $T^{*}Q$ has the expression
\begin{equation}
 X_H=\sum_{i=1}^n\left(\frac{\partial H}{\partial p_i}\frac{\partial}{\partial q^i}-\frac{\partial H}{\partial q^i}\frac{\partial}{\partial p_i}\right)
\end{equation}
on a $2n$ dimensional manifold.
\noindent
Its integral curves $(q^i(t),p_i(t))$ satisfy the Hamilton equations
\begin{equation}\label{hamileq22}
\left\{\begin{aligned}
 {\dot q}^i&=\frac{\partial H}{\partial p_i},\\
 {\dot p}_i&=-\frac{\partial H}{\partial q^i}
 \end{aligned}\right.
 \end{equation}
for all $i=1,\dots,n$. 

\begin{definition}
 Given two manifolds and a map between them, $F:M_1\rightarrow M_2$, we say that a vector field $X$ on $M_1$ and another vector
 field $Y$ on $M_2$ are $F$-related if
\begin{equation}
 Y{(F(x))}=dF(x)(X(x)), \qquad \text{for all}\quad x\in M_1 
\end{equation}

\end{definition}
\noindent
The Legendre transformation maps solutions of $\xi_L$ to solutions of $X_H$ since the Legendre transform is a symplectomorphism, that is $(FL)^{*}\omega_Q=\omega_L$.
Therefore, $\xi_L$ and $X_H$ are $FL$-related by the Legendre transformation.

\subsection{The Hamilton--Jacobi equation}
The {\it Hamilton--Jacobi equation} comes from the integral action along the solution over the time interval $(0,t)$
\begin{equation}\label{acc}
 S(q^i,t):=\int_0^t{\left(p_i(s)\dot{q}^i(s)-H(q^i(s),p_i(s))\right)ds}
\end{equation}
where the result is a function of the end point $(q,t)\in Q\times \mathbb{R}$. By taking variations of the end point, we arrive
at the {\it time-dependent Hamilton--Jacobi equation} \cite{Gold,Kibble}
\begin{equation}\label{tdepHJ}
 \frac{\partial S}{\partial t}+H\left(q^i,\frac{\partial S}{\partial q^i}\right)=0.
\end{equation}
Solving this equation consists on finding the {\it principal function} $S(q^i,t)$, where $H=H(q^i,p_i)$ is the Hamiltonian of the system. 
Conversely, it can be proven that if $S(q^i,t)$ is a solution of the Hamilton--Jacobi equation, then $S(q^i,t)$ is a generating function
for a family of symplectic flows that describe the dynamics of the Hamilton equations \eqref{hamileq22}.
 If the principal function is separable in time, then we can propose the {\it Ansatz} $S=W(q^1,\dots,q^n)-Et,$
where $E$ is the total energy of the system.

\noindent
Then, equation \eqref{tdepHJ} turns into

\begin{equation}\label{HJeq1}
 H\left({q}^i,\frac{\partial W}{\partial {q}^i}\right)=E.
\end{equation}
which is known as the {\it time-independent Hamilton--Jacobi equation}.
Indeed, if we find a solution $W$ of \eqref{HJeq1}, then any solution of the Hamilton
equations is retrieved by taking $p_i=\partial W/\partial {q}^i.$

Geometrically, this can be interpreted through a diagram (see below) in which a Hamiltonian vector field $X_{H}$ can be projected into the configuration manifold by means of a 1-form $dW$, and then the integral curves of the projected
vector field $X_{H}^{dW}$ can be transformed into integral curves of $X_{H}$ provided that $W$ is a solution of \eqref{HJeq1},
\[
\xymatrix{ T^{*}Q
\ar[dd]^{\pi} \ar[rrr]^{X_H}&   & &TT^{*}Q\ar[dd]^{T\pi}\\
  &  & &\\
 Q\ar@/^2pc/[uu]^{dW}\ar[rrr]^{X_H^{dW}}&  & & TQ}
\]
where

\begin{equation}
 X_H^{dW}= T_{\pi}\circ X_H\circ dW
\end{equation}

 This implies that $(dW)^{*}H=E$, with $dW$ being a section of the cotangent bundle. In other words, we are looking for a section $\alpha$ of $T^{*}Q$ 
such that $\alpha^{*}H=E$. As it is well-known, the image of a one-form is a Lagrangian submanifold of $(T^{*}Q, \omega_Q)$ if
and only if $d\alpha=0$ \cite{AbraMars,Arnold}.
That is, $\alpha$ is locally exact, say $\alpha=dW$ on an open subset around each point.

 Let $(T^{*}Q, \omega=-d\theta)$ be the cotangent bundle of $Q$ equipped with its canonical symplectic form $\omega_Q$, let $X_H$ be a Hamiltonian vector field
 on $T^{*}Q$ for a Hamiltonian $H$ and $X_H^{dW}$ a vector field on $Q$. Consider a function $W:Q\rightarrow \mathbb{R}$. The vector fields $X_H$ and $X_H^{dW}$ are $dW$-related
 if and only if
 \begin{equation}
  d(H\circ dW)=0.
 \end{equation}

\subsection{Discrete Mechanics}

Discrete Mechanics is a reformulation of the classical Lagrangian and Hamiltonian Mechanics with discrete variables. Its formulation appears from
discrete variational principles from which to derive analogues of the Euler--Lagrange (EL) and Hamilton equations in discrete form. There exist analogues of concepts of the continuous time framework. For example, we have symplectic forms, Legendre transformations, momentum maps
and Noether theorems \cite{OhsawaBlochLeok}.

Let $a,b\in \mathbb{R}$ and $a<b$, and $h=\frac{b-a}{N}$, where $N$ is the number of divisions of the discrete lattice where motion occurs. Consider $\mathbb{T}$ is a subspace of $\mathbb{R}$ defined by $\mathbb{T}=hZ\bigcap [a,b]$ where
$hZ=\{hz| z\in \mathbb{R}\}$.
Here, we denote by $C^{l}([a,b],\mathbb{R}^n)$ the set of $l$-times differentiable functions, for example $q:[a,b]\rightarrow \mathbb{R}^n$, this is $q\mapsto (q_1,\dots,q_n)$.
In the discrete framework, the Lagrangian is substituted by a discrete Lagrangian $L_d:Q\times Q\rightarrow \mathbb{R}$, where $Q$ is made of discrete variables $q \in C^{1}([a,b],\mathbb{R}^n)$. 
This discrete Lagrangian is an approximation of the {\it exact discrete Lagrangian} 
$$L_d^{ex}(q_j,q_{j+1})=\int_{t_j}^{t_{j+1}}{L(q(t),\dot{q}(t))dt},$$
where $q:[t_j,t_{j+1}]\rightarrow Q$ is the solution of the continuous EL equation with boundary conditions $q(t_j)=q_j,q(t_{j+1})=q_{j+1}$.

Now there exists a discrete Lagrangian flow in terms of points $\{q_j\}$\footnote{Notice that now the spatial coordinate has a subindex
that represents the discrete character of a single variable $q$, instead of the superindex which denotes one spatial coordinates of a set of $n$ different ones on a $n$ dimensional configuration space.} with $j=1,\dots,N$ on $Q$. 
The EL equations can be described by a discrete variational principle $\delta S_d=0$, where
\begin{equation}\label{acc2}
 S_d(\{q_j\})=\sum_{j=1}^{N-1} L_d(q_j,q_{j+1})
\end{equation}
with $j=1,\dots,N-1$.
In similar fashion as in Classical Mechanics, we can perform variations to derive the {discrete EL equations} in this case.
If we calculate $\delta S_d(q_j)=0$ with respect to a fixed point $q_j$, we obtain
\begin{equation}\label{disELE}
 D_2 L_d(q_{j-1},q_j)+D_1 L_d(q_j,q_{j+1})=0,
\end{equation}
where $D_1$ denotes partial derivative with respect to the first argument in the function $L_d$ and $D_2$ is the partial derivative with respect to the second argument.
Equations in \eqref{disELE} are known as the {\it discrete Euler--Lagrange equations} (DEL for short).

They give rise to a {\it Lagrangian discrete flow}
$\mathcal{F}_{L_d}:Q\times Q\rightarrow Q\times Q$ on the trivialized vector bundle $TQ\simeq Q\times Q$ such that
$$\mathcal{F}_{L_d}(q_{j-1},q_j)\rightarrow (q_j,q_{j+1}).$$
\noindent
Equivalently, we can define the {\it discrete one forms},
\begin{align}
 \theta^{+}_d&=D_2L_d(q_j,q_{j+1})dq_{j+1},\nonumber\\
 \theta^{-}_d&=-D_1L_d(q_j,q_{j+1})dq_{j}
\end{align}
that define a unique {\it discrete symplectic form}
\begin{equation}
 \Omega_d(q_j,q_{j+1})=-d\theta^{\pm}_d=-D_1D_2L_d(q_j,q_{j+1})\ dq_{j}\wedge dq_{j+1}
\end{equation}
and the flow $\mathcal{F}_{L_d}$ is a symplectomorphism, that is
$$\mathcal{F}_{L_d}^{*}\Omega_d=\Omega_d.$$

To derive a discrete Hamiltonian approach, we define {\it discrete Legendre transformations}, which are the {\it right} and {\it left} discrete Lagrange
transformations. Respectively,
\begin{align}\label{legtrans}
 \mathbb{F}L_d^{+}(q_j,q_{j+1})&=(q_{j+1},D_2L_d(q_j,q_{j+1})),\nonumber\\
\mathbb{F}L_d^{-}(q_j,q_{j+1})&=(q_{j},-D_1L_d(q_j,q_{j+1})),
\end{align}
for all $j=1,\dots,N-1.$ Generally, we will refer to the Legendre transformation (right $\mathbb{F}L_d^{+}$ or $\mathbb{F}L_d^{-}$, independently) as $FL$ simply.
From here, we can define the corresponding momenta as
\begin{align}
 p^{+}_{j,j+1}&=D_2L_d(q_j,q_{j+1}),\nonumber\\
p^{-}_{j,j+1}&=-D_1L_d(q_j,q_{j+1}).
\end{align}
which are normally unified under the common notation 
$$p_j:=p^{+}_{j-1,j}=p^{-}_{j,j+1},$$
due to the discrete Euler--Lagrange equations in \eqref{disELE}

The composition of the right discrete and left Legendre transforms is a flow
defined on the cotangent space $\mathcal{F}^H_d:T^{*}Q\rightarrow T^{*}Q$ 
\begin{equation}
 \mathcal{F}^H_d=\mathbb{F}L_d^{+}\circ (\mathbb{F}L_d^{-})^{-1}
\end{equation}

The following diagram summarizes the discrete Legendre transformations and their composition

\[
\xymatrix{ &Q\times Q
\ar[dl]^{\mathbb{F}L_d^{-}} \ar[dr]^{\mathbb{F}L_d^{+}}\ar[rr]^{\mathcal{F}_{L_d}}&  & Q\times Q\ar[dl]^{\mathbb{F}L_d^{-}} \ar[dr]^{\mathbb{F}L_d^{+}} &  \\
 T^{*}Q\ar[rr]^{\mathcal{F}^H_d}& & T^{*}Q\ar[rr]^{\mathcal{F}^H_d}& & T^{*}Q}
\]

Point to point, 

\[
\xymatrix{ &(q_{j-1},q_{j})
\ar[dl]^{\mathbb{F}L_d^{-}} \ar[dr]^{\mathbb{F}L_d^{+}}\ar[rr]^{\mathcal{F}_{L_d}}&  & (q_{j},q_{j+1})\ar[dl]^{\mathbb{F}L_d^{-}} \ar[dr]^{\mathbb{F}L_d^{+}} &  \\
 (q_{j-1},p_{j-1})\ar[rr]^{\mathcal{F}^H_d}& & (q_{j},p_{j})\ar[rr]^{\mathcal{F}^H_d}& & (q_{j+1},p_{j+1})}
\]

\noindent
The {\it discrete Hamiltonian flow} is $\mathcal{F}_d^{H}:T^{*}Q\rightarrow T^{*}Q$ is a symplectomorphism, that is $(\mathcal{F}_{d}^{H})^{*}\omega_Q=\omega_Q$
that brings points into points 
$$\mathcal{F}_{d}^{H}:(q_j,p_j)\rightarrow (q_{j+1},p_{j+1})$$

To derive a Hamiltonian formalism, we use that a discrete Lagrangian is essentially a generating function of type one \cite{Arnold} and that
we can apply the defined Legendre transformations to the discrete Lagrangian to find a discrete Hamiltonian \cite{Arnold,Gold2}.
With the {right} Legendre transformation, we have
\begin{equation}\label{pd2}
 p_{j+1}=D_2L_d(q_j,q_{j+1}).
\end{equation}
Here we perform local computations. We can identify the configuration manifold $Q$ with $\mathbb{R}^n$, then we can define a {\it discrete Hamiltonian} as the function $H:\mathbb{R}^n\times \mathbb{R}^n\rightarrow \mathbb{R}$ such that for $(q,p)\in C^{1}(\mathbb{T},\mathbb{R}^n)\times C^{1}(\mathbb{T},\mathbb{R}^n)$
we have time evolution of $(q,p)$ given by the discrete Hamilton equations. We define the {\it right discrete Hamiltonian}
\begin{equation}\label{rdh}
 H^{+}_d(q_j,p_{j+1})=p_{j+1}q_{j+1}-L_d(q_j,q_{j+1})
\end{equation}
and we obtain the {\it right discrete Hamilton equations}
\begin{align}\label{rdhe}
 q_{j+1}&=D_2H^{+}_d(q_j,p_{j+1}),\nonumber\\
p_j&=D_1H^{+}_d(q_j,p_{j+1}).
\end{align}
Equivalently, with the {left} Legendre transformation, we can obtain the {\it left discrete Hamiltonian}
$$H^{-}_d=-p_jq_j-L_d(q_j,q_{j+1})$$
and the {\it left discrete Hamilton equations}
\begin{align}\label{ldhe}
 q_{j}&=-D_2H^{-}_d(q_{j+1},p_j),\nonumber\\
p_{j+1}&=-D_1H^{-}_d(q_{j+1},p_j).
\end{align}
\noindent
{\it Remark}: There exists a discrete version of the extended Hamilton's variational principle \cite{Gold}. It says
 \begin{theorem}
  The points $(q,p)\in C^{1}([a,b],\mathbb{R}^n)\times C^{1}([a,b],\mathbb{R}^n)$ satisfying the discrete Hamilton equations are critical points of the functional
 \begin{equation}
  \mathcal{L}_H: C^{1}([a,b],\mathbb{R}^n)\times C^{1}([a,b],\mathbb{R}^n)\longrightarrow \mathbb{R}
 \end{equation}
 such that
 \begin{equation}
  \mathcal{L}_H(q,p)=\int_a^b{L_H(q(t),\dot{q}(t),p(t),\dot{p}(t))dt}
 \end{equation}
 where 
 $$L_H=p_{j+1}q_{j+1}-H^{+}(q_j,q_{j+1}).$$
 \end{theorem}

\subsection{The discrete Hamilton--Jacobi equation}

The discrete Hamiltonian theory and in particular, the discrete Hamilton--Jacobi equation were developed as a generalization
of nonsingular, discrete optimal control problems \cite{Lall}. The discrete Hamilton--Jacobi equation is expected as
the outcome of a discrete variational problem. 
If we reconsider the discrete action \eqref{acc2},
\begin{equation*}
 S_d^N(\{q_j\}_{j=1,\dots,N})=\sum_{j=0}^{N-1} L_d(q_j,q_{j+1})
\end{equation*}
that written in terms of the right discrete Hamiltonian \eqref{rdh},
\begin{equation*}
 S_d^j(q_j)=\sum_{k=1}^{j-1} \left(p_{k+1}q_{k+1}-H_d^{+}(q_k,p_{k+1})\right)
\end{equation*}
which if evaluated along the solution of the right discrete Hamilton equations \eqref{rdhe}, then $S_d^{j}(q_j)$ is a function
of the end point coordinates $q_j$ and the discrete end time $j$.

On the other hand, some previous works \cite{ElnaSch} have specifically derived an equation
based on the philosophy of a generating function of a coordinate transformation that trivializes the dynamics \cite{Gold,Gold2}.
The work by T. Oshawa, A.M. Bloch and M. Leok \cite{OhsawaBlochLeok} generalizes the previous statement by finding
a discrete generating function $S^{j}(q_j)$ of a transformation $(q_j,p_j)\rightarrow (Q_j,P_j)$ in which the discrete dynamics is trivial.
The main theorem is the following.
\begin{theorem}
 Consider the right discrete Hamilton equations \eqref{rdhe} and a discrete phase space $\{(q_j,p_j)\}_{j=1}^N$. Consider a change of coordinates
 $(q_j,p_j)\rightarrow (Q_j,P_j)$, for all $j=1,\dots,N$ that satisfies
 \begin{enumerate}
  \item The old and new coordinates are related by a generating function $S^j:\mathbb{R}^n\rightarrow \mathbb{R}$
  of the type
  \begin{align}
   P_j&=-D_1S^j(Q_j,q_j),\nonumber\\
   p_j&=D_2S^j(Q_j,q_j).
  \end{align}
\item The dynamics in the new coordinates $\{(Q_j,P_j)\}_{j=1}^N$ is rendered trivial, i.e., $(Q_{j+1},P_{j+1})=(Q_j,P_j).$
  \end{enumerate}
  Then, the set of functions $\{S^j_d\}$ with $j=1,\dots,N$ satisfies the discrete Hamilton--Jacobi equation:
  \begin{align}\label{firstdiscrete}
    S_d^{j+1}(q_{j+1})&-S_d^j(q_j)\nonumber\\
&-DS_d^{j+1}(q_{j+1})q_{j+1}+H_d^{+}(q_j,DS^{j+1}_d(q_{j+1}))=0.
  \end{align}

\end{theorem}

See reference \cite{OhsawaBlochLeok} for proof of this theorem.

\section{A geometric and discrete Hamilton--Jacobi theory}

In this section we obtain a discrete geometric Hamilton--Jacobi theory in terms of projected flows and projected Hamiltonian vector fields\footnote{By projected we do not refer to a projective flow/vector field but to the restriction of a Hamiltonian flow/vector field on the phase space $T^{*}Q$ 
along the image of a Lagrangian submanifold $dW$.}.
In particular, the problem of a discrete theory in terms of vector fields roots in the definition of a discrete vector field, that we introduce
in forthcoming subsections. 


%
%

\subsection{The discrete flow approach}

A different approach but equivalent to the usual Hamilton--Jacobi theory relying on the projection of a Hamiltonian vector field
via $\gamma=dW$ is here substituted by the projection of discrete flows. 

We propose an analogue for the geometric diagram as follows \cite{LeonSardon1,LeonSardon2}. Consider the discrete flow $\mathcal{F}_d^{H}:T^{*}Q\rightarrow T^{*}Q$
and a discrete section $\gamma=DS_d$, where $S_d:Q\rightarrow \mathbb{R},$ is the discrete generating function. The projected
flow is here $(\mathcal{F}_d^{H})^{DS_d}:Q\rightarrow Q$.

\[
\xymatrix{T^{*}Q
\ar[dd]^{\pi_Q} \ar[rrr]^{\mathcal{F}_d^{H}}&   & &  T^{*}Q\ar[dd]^{\pi_Q}\\
  &  & &\\
 Q\ar@/^2pc/[uu]^{DS_d}\ar[rrr]^{(\mathcal{F}_d^{H})^{DS_d}}&  & & Q\ar@/^2pc/[uu]^{DS_d}}
\]

The point to point interpretation is
\[
\xymatrix{ (q_j,p_{j})
\ar[dd]^{\pi_Q} \ar[rrr]^{\mathcal{F}_d^{H}}&   & & (q_{j+1},p_{j+1})\ar[dd]^{\pi_{Q}}\\
  &  & &\\
 (q_j)\ar@/^2pc/[uu]^{DS_d^j(q_j)}\ar[rrr]^{(\mathcal{F}_d^{H})^{DS_d}}&  & & (q_{j+1})\ar@/^2pc/[uu]^{DS_d^{j+1}(q_{j+1})}}
\]
where $\pi_Q$ is the natural projection to the configuration manifold and the flow is such that 
$$(\mathcal{F}_d^{H}):(q_j,DS_d^j(q_j))\longrightarrow (q_{j+1},DS_d^{j+1}(q_{j+1})).$$
Here $\{S_d^j\}$ is a family of generating functions of the Hamilton--Jacobi equation.
\noindent
We say that two flows are related if the following condition is fulfilled
\begin{equation}\label{comp22}
 (\mathcal{F}_d^H)^{DS_d}=\pi_{Q}\circ \mathcal{F}_d^H\circ DS^j(q_j).
\end{equation}
\noindent
This is equivalent to saying that, point to point,
\begin{equation*}
 \mathcal{F}_d^{H}(q_j,DS_d(q_j))=(q_{j+1},DS_d(q_{j+1}))
\end{equation*}

This is key to the following theorem.

\begin{theorem}[The discrete Hamilton--Jacobi theorem]
The two flows \newline $(\mathcal{F}_d^{H})^{DS_d(q_j)}$ and $\mathcal{F}_d^{H}$ are $DS_d$-related if the following equation
\begin{align}\label{hjen}
 S_d^{j+1}(q_{j+1})&-S_d^j(q_j)\nonumber\\
&-DS_d^{j+1}(q_{j+1})q_{j+1}+H_d^{+}(q_j,DS^{j+1}_d(q_{j+1}))=0
\end{align}
is satisfied. We shall refer to \eqref{hjen} as the {\it discrete Hamilton--Jacobi equation}.
\end{theorem}
Then, we say that $DS_d$ is a {\it discrete solution} for the {discrete Hamilton--Jacobi equation} and $S_d$ is the generating
function.
\begin{proof}
 Considering the definition of the action in \eqref{acc2} and the right Legendre transform \eqref{legtrans}, we have that
\begin{equation}\label{hje2}
  S_d^{j+1}(q_{j+1})-S_d^j(q_j)=p_{j+1}q_{j+1}-H_d^{+}(q_j,p_{j+1}).
\end{equation}
If we derivate with respect to $q_{j+1}$, we obtain
\begin{equation}\label{intro}
 p_{j+1}=DS^{j+1}_d(q_{j+1}),
\end{equation}
and considering the right discrete Hamilton equations, in which $q_{j+1}=D_2H_d^{+}(q_j,p_{j+1})$, if we introduce 
\eqref{intro} into \eqref{hje2}, we arrive at \eqref{hjen}, which is the discrete Hamilton--Jacobi equation.
On the other hand, the flow interpretation using \eqref{comp22} provides
\begin{align}
 &\mathcal{F}_d^H\circ DS_d^j(q_j)=\mathcal{F}_d^H(q_j,p_j)=(q_{j+1},D_2L(q_j,q_{j+1})), \quad \text{and}\\
 &DS_d^{j+1}\circ (\mathcal{F}_d^H)^{DS_d}(q_j)=DS_d(q_{j+1})=(q_{j+1},DS_d(q_{j+1})).
\end{align}
From the commutativity of the diagram, we have
\begin{equation}
 D_2L(q_j,q_{j+1})=DS_d^{j+1}(q_{j+1}).
\end{equation}
that means
\begin{equation}
 D_2L(q_j,q_{j+1}) q_{j+1}=H_d^{+}(q_{j},p_{j+1})+S_d^{j+1}(q_{j+1})-S_d^{j}(q_j)
\end{equation}
according to \eqref{hjen}, and necessarily
\begin{equation}
 p_{j+1}=D_2L_d(q_j,q_{j+1})
\end{equation}
which is true due to definition \eqref{pd2}.

\end{proof}
There is an equivalent interpretion of the equation in terms of the left discrete action. See Appendix A.

%
%
%

\subsection{The discrete vector field approach}

According to the usual geometric Hamilton--Jacobi theory constructed out of vector fields,
analogously to the continuous case, we introduce a commutative diagram for the discrete case based on the results
of discrete Hamiltonian vector fields introduced by Cresson and Pierret \cite{cresson}. 
The discrete least action principle (DLAP for short) worked upon a discrete Lagrangian gives rise to the discrete Euler--Lagrange equations. A discrete Hamilton 
gives rise to a {\it discrete Hamiltonian vector field} $X_d$.

\[
\xymatrix{ L_d
\ar[dd]^{\text {DLAP}} \ar[rrr]^{\text{Leg. transf}}&   & &H\ar[dd]^{\text{definition}}\\
  &  & &\\
 \text{DEL}\ar[rrr]&  & & X_d}
\]

Using the right discrete Hamiltonian \eqref{rdh}, we define its corresponding right discrete Hamilton equations \eqref{rdhe},
and the right discrete Hamiltonian vector field reads \cite{cresson}

\begin{equation}\label{dhvf}
 X_d=\sum_{j=1}^{N-1} \left( D_2H^{+}(q_j,p_{j+1})\frac{\partial}{\partial q_{j+1}}+D_1H^{+}(q_j,p_{j+1})\frac{\partial}{\partial p_j}\right)
\end{equation}

Equivalently, a left discrete Hamiltonian vector field can be defined and the theory can be reconstructed in terms of it (see appendix B).

\medskip

We propose the following commutative diagram for a discrete Hamilton--Jacobi formulation in terms of discrete vector fields,
where $S_d:\mathbb{R}^n\rightarrow \mathbb{R}$ and the vertical arrows denote the obvious projections. We consider
the cotangent bundle $T^{*}Q$ and suppose that $Q$ is locally diffeomorphic to $\mathbb{R}^n.$ Of course, this
would be the case because we are performing local computations.

\[
\xymatrix{ \mathbb{R}^n\times \mathbb{R}^n
\ar[dd]^{\pi_{\mathbb{R}^n}} \ar[rrr]^{X_d}&   & &T(\mathbb{R}^n\times \mathbb{R}^n)\ar[dd]^{T\pi_{\mathbb{R}^n}}\\
  &  & &\\
 \mathbb{R}^n\ar@/^2pc/[uu]^{\gamma=DS_d}\ar[rrr]^{X_d^{DS_d}}&  & & T\mathbb{R}^n}
\]

\begin{definition}
We define the projected vector field $X_d^{DS_d}:\mathbb{R}^n \rightarrow T\mathbb{R}^n$ depicted in the diagram above,
in the following way
\begin{equation}\label{comp}
 X_d^{DS_d}=T_{\pi_{\mathbb{R}^n}}\circ X_d\circ DS_d
\end{equation}
so that the diagram is commutative.
\end{definition}

\begin{theorem}[The discrete Hamilton--Jacobi theorem]
The discrete vector fields $X_d$ and $X_d^{DS_d}$ are $DS_d$-related 
if the following equation is satisfied
\begin{equation}\label{midhje}
 D_2H^{+}(q_j,p_{j+1})D_{q_j}\gamma_j(q_{j+1})=D_1H^{+}(q_{j},p_{j+1})
\end{equation}
where $\gamma=DS_d$.
If the two discrete vector fields are $DS_d$-related or $\gamma$-related, equivalently, we can say that $DS_d$ maps integral curves of $X_d^{DS_d}$ into solutions
of $X_d$, that is, solutions of the Hamilton equations.

\begin{proof}
In order for \eqref{comp} to be satisfied, we have to perform the calculation
\begin{equation}
 T_{\gamma} X_d^{DS_d}=X_d,\qquad \gamma=DS_d
\end{equation}
We look for a section $\gamma=\{\gamma_j(q_{j+1}), \forall j=1,\dots,N-1\}$ such that
{\begin{footnotesize}
\begin{align}
 T_{\gamma} X_d^{DS_d}&=T_{\gamma} \left(\sum_{j=1}^{N-1}  D_2H^{+}(q_j,p_{j+1})\frac{\partial}{\partial q_{j+1}}\right)= \sum_{j=1}^{N-1} D_2H^{+}(q_j,p_{j+1})T_{\gamma} \left(\frac{\partial}{\partial q_{j+1}}\right)\nonumber\\
&=\sum_{j=1}^{N-1} D_2H^{+}(q_j,p_{j+1})\left(\frac{\partial}{\partial q_{j+1}}+D_{q_{j+1}}\gamma_j(q_{j+1})\frac{\partial}{\partial p_{j+1}}\right)
\end{align}
\end{footnotesize}}
which has to be equal to \eqref{dhvf}. From this, we obtain the expression
\begin{equation}
 D_2H^{+}(q_j,p_{j+1})D_{q_{j+1}}\gamma_j(q_{j+1})-D_1H^{+}(q_j,p_{j+1})=0
\end{equation}
that is another way of describing the {\it discrete Hamilton--Jacobi equation}
\begin{equation}\label{anotherhjen}
 D_2H^{+}_dD\gamma=D_1H^{+}_d.
\end{equation}
\end{proof}
\end{theorem}
Note: The left discrete formulation leads to equivalent results.

\begin{proposition}
The discrete flow formulation and the discrete vector field approach for the discrete, geometric Hamilton--Jacobi equation are equivalent.
\end{proposition}
\begin{proof}
 The proof is straightforward. It consists of taking the total derivative of expression \eqref{hjen}
 and considering $q_{j+1}=D_2H^{+}(q_j,p_{j+1})$ from the right discrete Hamilton equations. As a byproduct
 we obtain two copies of the same expression, that corresponds with \eqref{anotherhjen}.
 
\end{proof}

\section{Applications}

In \cite{Sakamoto} the authors propose two
approximation methods to solve optimal control problems: the Hamiltonian perturbation technique and the stable manifold approach.
Here, we propose the use of discrete Hamilton--Jacobi equations as an alternative and third  method to obtain approximate solutions of optimal control problems.
We can compare the power of our approach by comparising our results with the two proposed approaches in \cite{Sakamoto}.

\begin{definition}
A control problem of  ordinary differential equations is usually given by
\begin{equation}\label{control}
 \dot{q}^i=\Gamma^i(q(t),u(t)),\quad 1\leq i\leq n
\end{equation}
where $\{q^i\}$ are called state variables and $\{u^a\}, 1 \leq a\leq k$ are control functions.
\end{definition}
The optimal control is the following. Given initial and final states $q_0$ and $q_f$, the objective is to find a $C^2$ piecewise
curve $c(t)=(q(t),u(t))$ such that $q(t_0)=q_0$ and $q(t_f)=q_f$, satisfying the control equations and minimizing the functional
\begin{equation*}
 \mathcal{J}(c)=\int_{t_0}^{t_f} \mathbb{L}(q(t),u(t))dt
\end{equation*}
\noindent
for some cost function $\mathbb{L}=\mathbb{L}(q,u)$. 

%
%
%

For a geometrical description, one assumes a fiber bundle structure $\pi:C\rightarrow B$, where $B$ is the configuration manifold
with local coordinates $\{q^i\}$ and $C$ is the bundle of controls with local coordinates $\{q^i,u^a\}$. The ordinary differential
equations in \eqref{control} on $B$ depending on the parameters $u$ can be seen as a vector field $\Gamma$ along the projection map $\pi$
that is, $\Gamma$ is a smooth map $\Gamma:C\rightarrow TB$ such that the following diagram is commutative.


\begin{center}
 \begin{tikzcd}[column sep=huge, row sep=huge]
  &\mathbb{R}  &C\arrow[d, "\pi"] \arrow[r, "\Gamma"]\arrow[l, "\mathbb{L}"] &TB\arrow[dl, "\tau_{B}"]   \\
  & &B &
 \end{tikzcd}
\end{center}

The dynamics is here restricted to a submanifold given the restrictions of the control equations \eqref{control}.

\begin{center}
 \begin{tikzcd}[column sep=huge, row sep=huge]
  &M \arrow[r, hook, "i"]  &TC\arrow[r, "\phi"]\arrow[d, "\tau_{C}"] &\mathbb{R}   \\
  & &C\arrow[ur, "\mathbb{L}"] &
 \end{tikzcd}
\end{center}



So, the optimal control problem $(C,\mathbb{L},\Gamma)$ is associated with the Lagrangian function
$L:TC\rightarrow \mathbb{R}$, where $L=\mathbb{L}\circ \tau_C$ and the constraint submanifold $M$ defined by
\begin{equation}
 M=\{(q^i,u_a,\dot{q}^i,\dot{u}_a)\quad |\quad \dot{q}^i=\Gamma^i(q^i,u_a)\}
\end{equation}
for bundle coordinates $\{q^i,u_a,\dot{q}^i,\dot{u}_a\}$ on $TC$, then $L=L(q^i,u_a,\dot{q}^i,\dot{u}_a)$ and $\mathbb{L}=\mathbb{L}(q^i,u_a)$.
%
%

Let us define a singular Lagrangian $\mathcal{L}:T(C\times \mathbb{R}^n)\rightarrow \mathbb{R}$ in terms of Lagrange multipliers $\lambda_i$ \cite{Arnoldiii},
\begin{equation}
 \mathcal{L}=\mathbb{L}+\lambda_i(\dot{q}^i-\Gamma^{i}(q^i,u_a))
\end{equation}
and the Legendre transformation $FL$ of this Lagrangian

\begin{center}
 \begin{tikzcd}[column sep=huge, row sep=huge]
  &T(C\times \mathbb{R}^n) \arrow[r, "FL"]  \arrow[dr, "FL_1"] &T^{*}(C\times \mathbb{R}^n)   \\
  & &M_1 \arrow[u, hook, "i"] 
 \end{tikzcd}
\end{center}
where $FL_1$ is the restriction of the Legendre transformation to the {\it first-order constraint submanifold} $M_1$.

Now, we apply the Dirac-Bergmann algorithm \cite{dirac} geometrized by M. Gotay and J.M. Nester \cite{GN,GNH}. In
bundle coordinates $(q^i,u_a,\lambda_i,\dot{q}^i, \dot{u}_a, \dot{\lambda}_i)$ on $T^{*}(C\times \mathbb{R}^n)$, the first-order constraint submanifold
$M_1$ is locally defined by the implicit equations
\begin{equation}\label{impeq}
\left(q^i,u_a,\lambda_i,\frac{\partial \mathcal{L}}{\partial q^i}=\lambda_i,\frac{\partial \mathcal{L}}{\partial u_a}=0, \frac{\partial \mathcal{L}}{\partial \lambda_i}=0\right)
\end{equation}
on $T^{*}(C\times \mathbb{R}^n)$.
The definition of the energy function is
\begin{equation}
 E_{\mathcal{L}}=\dot{q}^i\lambda_i-\mathcal{L}=\lambda_i\Gamma^i(q^i,u_a)-\mathbb{L}(q^i,u_a)
\end{equation}

Here, $E_{\mathcal{L}}$ constant along the fibers of $FL_1$ and projects to $M_1$. For this we say that $\mathcal{L}$ is {\it almost regular}.
Hence, the constrained Hamiltonian is
\begin{equation}\label{elh1}
 H_1(q^i,u_a,\lambda_i)=\lambda_i\Gamma^i-\mathbb{L}
\end{equation}
The symplectic form on $T^{*}(C\times \mathbb{R}^n)$ is
\begin{equation}
 \omega_{T^{*}(C\times \mathbb{R}^n)}=dq^i\wedge dp_{q^i}+du_a\wedge dp_{u_a}+d\lambda_i\wedge dp_{\lambda_i}
\end{equation}
and then, its restriction to $M_1$ is
\begin{equation}
 \omega_1=\omega_{C\times \mathbb{R}^n}|_{M_1}=dq^i\wedge d\lambda_i
\end{equation}
and the vector field $X_1$ providing the dynamics on $M_1$ will fulfill
\begin{equation}
 \iota_{X_1}\omega_1=dH_1
\end{equation}
It reads,
{\begin{footnotesize}
\begin{equation*}
 X_1=-\Gamma^i \frac{\partial}{\partial q^i}+\left(\lambda_i\frac{\partial \Gamma}{\partial q^i}-\frac{\partial \mathbb{L}}{\partial q^i}\right)\frac{\partial}{\partial \lambda_i}
\end{equation*}
\end{footnotesize}}
\noindent
from where we obtain restrictions that define the {\it secondary constraint manifold} $M_2$,
\begin{equation}\label{phia}
 \phi_a=\lambda_i\frac{\partial \Gamma^i}{\partial u^a}-\frac{\partial \mathbb{L}}{\partial u^a}=0
\end{equation}
which are called {\it secondary constraints}. Furthermore, the tangency condition $X_1(\phi_a)=0$ provides the regularity condition
we assume for optimal control problems.

\begin{example}[A one dimensional nonlinear control problem]\normalfont
Consider a one dimensional nonlinear control problem \cite{Sakamoto} whose continuous version is
\begin{align}
 \dot{q}&=q-q^3+u,\\
 \mathcal{J}&=\int_{0}^{\infty}\left(\frac{s}{2}q^2+\frac{r}{2}u^2\right)dt
\end{align}
and whose restricted Hamiltonian according to the algorithm described above (but
using the opposite sign criterion in order to retrieve results exposed in \cite{Sakamoto} where they use the positive sign) is
\begin{equation}\label{controlham}
H=p(q-q^3)-\frac{1}{2r}p^2+\frac{s}{2}q^2.
 \end{equation}
The constraint \eqref{phia} is

 \begin{equation}\label{phia}
 \phi_a=\lambda_i\frac{\partial \Gamma^i}{\partial u^a}+\frac{\partial \mathbb{L}}{\partial u^a}=0
\end{equation} 
with the positive sign criterion in \cite{Sakamoto}. For this one dimensional nonlinear control problem,
 \begin{equation}
  \phi_a=p+ru=0
 \end{equation}

and the vector field $\Gamma$ reads
 \begin{equation}
  \Gamma=(q-q^3+u)\frac{\partial}{\partial q}.
 \end{equation}

 In the discrete case, the right discrete Hamiltonian would read
 \begin{equation}\label{disopham}
  H_d^{+}=(q_j-q_j^3)p_{j+1}-\frac{p_{j+1}^2}{2r}+\frac{s}{2}q_j^2
 \end{equation}
\noindent
 So, the associated right discrete Hamilton equations are
 \begin{align}\label{optcontrolrdhe}
  q_{j+1}&=q_j-q_j^3-\frac{1}{r}p_{j+1},\nonumber\\
  p_j&=(1-3q_j^2)p_{j+1}+sq_j.
 \end{align}

As a matter of simplicity let us choose the parameters $r=s=1$, without loss of generalization.
The orbits in the discrete phase space take the form

 \begin{figure}[H]
\centering
\begin{minipage}[l]{6cm}
\includegraphics[width=6cm,height=5cm]{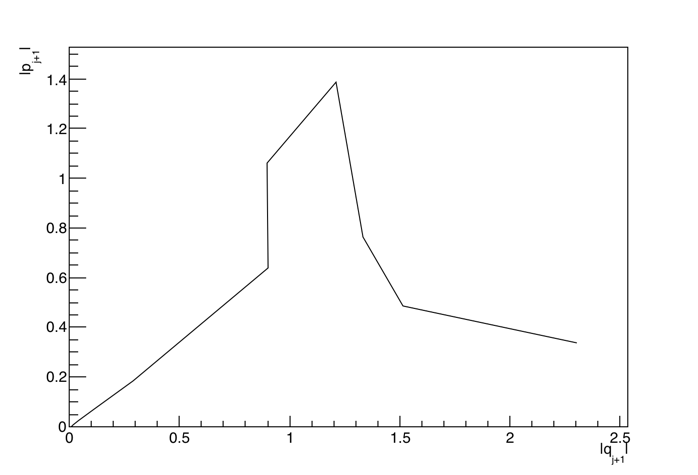}
\end{minipage}
\begin{minipage}[r]{6cm}
\includegraphics[width=6cm,height=5cm]{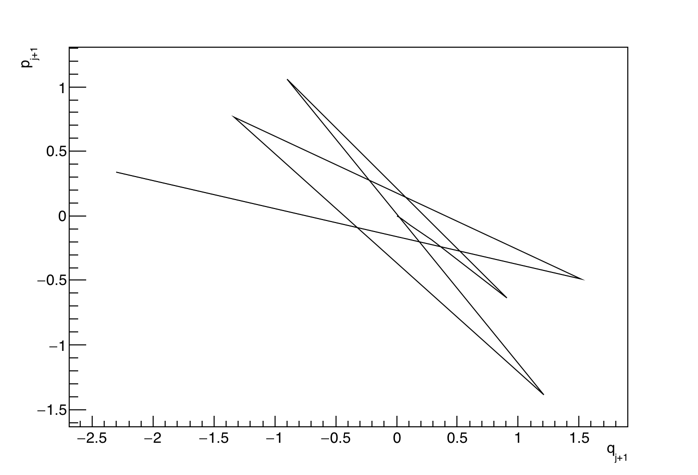}
\end{minipage}
\caption{$|p_{j+1}|$ vs. $|q_{j+1}|$ and $p_{j+1}$ vs. $q_{j+1}$}\label{figures}
\end{figure}

\noindent
for values $q_1=0.00000005$ and $p_1=0$, which is compatible with results given in \cite{Sakamoto} for a continuous version. It is easy to see
 that the curve in \cite{Sakamoto} is an equivalent continuous version of our representation above.
 This could be reenacted in terms of the left discrete Hamiltonian.
 
 \subsubsection*{The discrete flow approach}
 
 To obtain a result of the Hamilton--Jacobi
 equation applied to our optimal control problem, we need to solve the generating function $S_d^j$ or equivalently, $DS_d^j$.
 For this, we use equation \eqref{hjen}, whose solution for this particular example is
 
 \begin{align}\label{ecfordse}
  DS^{j+1}_d=-&q_j^3+q_j-q_{j+1}\pm \nonumber \\
   &\sqrt{q_j^6-2q_j^4+2q_j^3q_{j+1}+2hDS^j_d+2q_j^2-2q_jq_{j+1}+q_{j+1}^2}
 \end{align}
 
 Solving recurrently this expression for initial values $q_{1}=0.00000005,q_2=1.5*10^{-7}$ and $DS^1_d=0$ and a value $h=0.0001$, we obtain a graphic for
 values of $DS_d^j$ versus $q_j$ and for the absolute values of $|DS_d^j|$ versus $|q_j|$.

  \begin{figure}[H]
\centering
\begin{minipage}[l]{6cm}
\includegraphics[width=6cm,height=5cm]{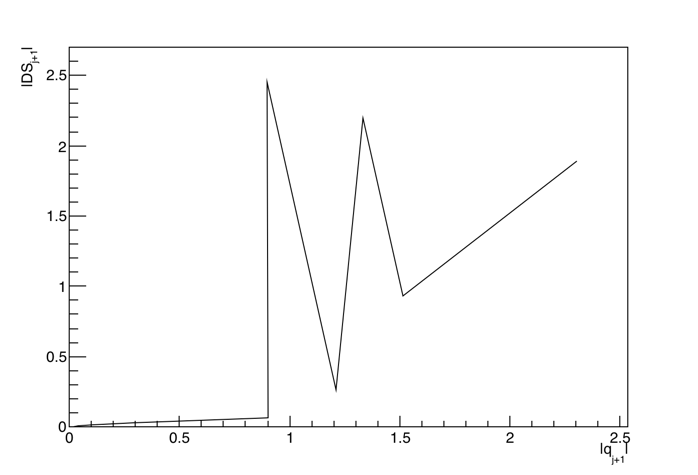}
\end{minipage}
\begin{minipage}[r]{6cm}
\includegraphics[width=6cm,height=5cm]{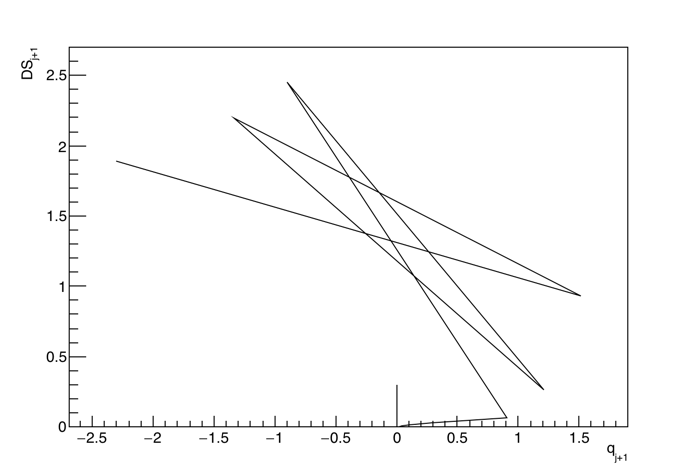}
\end{minipage}
\caption{$|DS^{j+1}_d|$ vs. $|q_{j+1}|$ and $DS^{j+1}_d$ vs. $q_{j+1}$ }
\end{figure}
 
 The graphic on the right hand side shows that the phase space obtained for $(DS^{j+1}_d,q_{j+1})$ where
 $DS^{j+1}_d$ plays the role of $p_{j+1}$ is equivalent to the phase space $(p_{j+1},q_{j+1})$ given by the
 right discrete Hamilton equations \eqref{optcontrolrdhe}. Indeed, the form and variation of the variable are the same but
 there is a displacement along the $y$ axis because of constant terms in \eqref{ecfordse} that produce this shift.
 
 The graphic on the left hand side shows a similar behavior between the absolute value phase space $(|DS^{j+1}_d|,|q_{j+1}|)$
 and the absolute value phase space  $(|p_{j+1}|,|q_{j+1}|)$ from \eqref{optcontrolrdhe}. Indeed, there is a linear growth of $|DS^{j+1}_d|$ and $|p_{j+1}|$ 
 between values $|q_{j+1}|=\{0,0.9\}$ and a peak around $|q_{j+1}|=1$. The discordance between both graphics is rooted
 in the $y$ axis shift commented for the case on right hand side.
 
 This means that although it is evident that $DS^{j+1}_d$ obtained from equation \eqref{ecfordse} by \eqref{hjen} and $p_{j+1}$ are
 equivalent, given the representations $DS^{j+1}_d$ vs. $q_{j+1}$ and $p_{j+1}$ vs. $q_{j+1}$, the phase shift in the $y$ axis
 is quite visible in the absolute value phase space.
 
 The next subsection shows that the results obtained through the discrete vector field approach are more accurate and there is
 no axis shift.

\subsubsection*{The discrete vector field approach}

To apply the discrete vector field approach in our optimal control problem, we need to impose condition \eqref{comp} for a vector
field that reads

\begin{equation}
 X_d=\sum_{j=1}^{N-1}\left(q_j-q_j^3-p_{j+1}\right)\frac{\partial}{\partial q_{j+1}}+\left((1-3q_j^2)p_{j+1}+q_j\right)\frac{\partial}{\partial p_j}
\end{equation}
\noindent
and whose projection is

\begin{equation}
 X_d^{\gamma}=\sum_{j=1}^{N-1}\left(q_j-q_j^3-p_{j+1}\right)\frac{\partial}{\partial q_{j+1}}
\end{equation}
\noindent
We choose a section $\gamma=\{\gamma_j(q_{j+1}), \forall j=1,\dots,N-1\}$ and through \eqref{comp}, we obtain the following equation,
\begin{equation}
 \left(q_j-q_j^3-p_{j+1}\right)\frac{\partial \gamma_j(q_{j+1})}{\partial q_{j+1}}=\left(1-3q_j^2\right)p_{j+1}+q_j
\end{equation}

whose solution is

\begin{equation}\label{ginit}
 \gamma_{j+1}(q_{j+2})=\frac{-\left(\gamma_j(q_{j+1})q_j^2-\gamma_j(q_{j+1})+q_{j+1}\right)q_j}{\gamma_j(q_{j+1})+q_{j+1}-3q_j^2q_{j+1}}
\end{equation}

Solving this expression by imposing initial values $\gamma_{1}=0$, $q_1=0.00000005$ and $q_{2}=0.00000005$, we obtain the following values
if we represent $|\gamma_{j+1}(x_{j+2})|$ vs. $|q_{j+1}|$ and $\gamma_{j+1}(q_{j+2})$ vs. $q_{j+1}$, we have

\medskip

  \begin{figure}[H]
\centering
\begin{minipage}[l]{6cm}
\includegraphics[width=6cm,height=5cm]{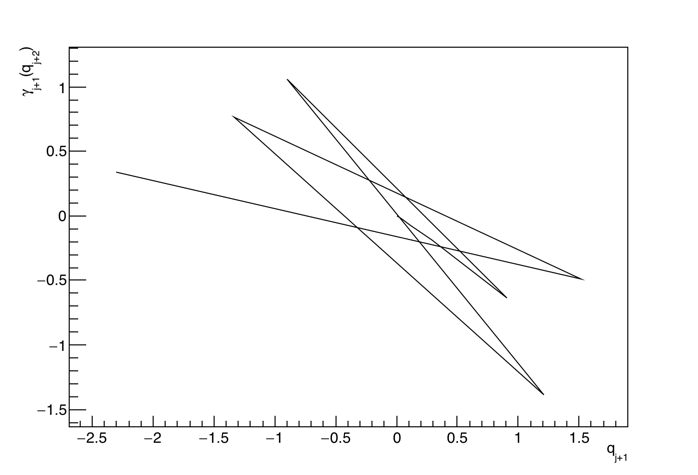}
\end{minipage}
\begin{minipage}[r]{6cm}
\includegraphics[width=6cm,height=5cm]{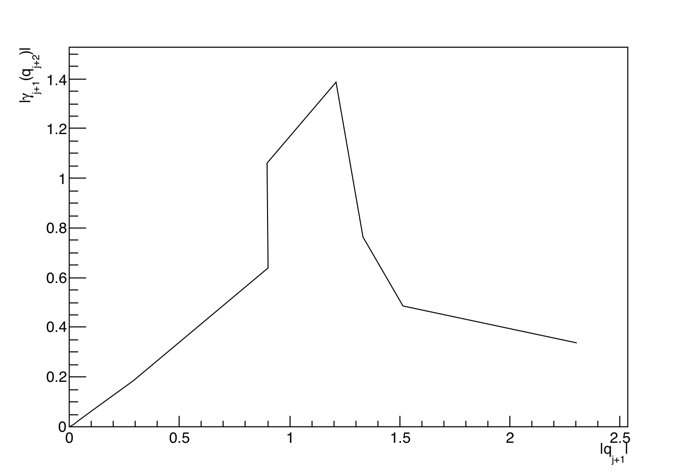}
\end{minipage}
\caption{$|\gamma_{j+1}|$ vs. $|q_{j+1}|$ and $\gamma_{j+1}$ vs. $q_{j+1}$.}\label{figures}
\end{figure}

From these graphics, we can clearly see that there is a good match between the results obtained for $\gamma_{j+1}$ playing the
role of the momenta $p_{j+1}$ and the momenta themselves $p_{j+1}$ of the phase space \eqref{optcontrolrdhe}. There exists no phase shift in the $y$ axis as it happened in
the discrete flow interpretation.

\subsubsection*{Comparion of methods}

From the previous graphics, it is clear that the discrete vector field interpretation seems more accurante than
the discrete flow interpretation and the discrete generating function formula \eqref{hjen}.
To see the accuracy of the discrete Hamiltonian vector field approach, we represent the matching between $\gamma_j(q_{j+1})$
representing the role of $|p_{j+1}|$ and $|p_{j+1}|$ from \eqref{optcontrolrdhe}.

 \begin{figure}[H]
 \centering
 \includegraphics[width=90mm]{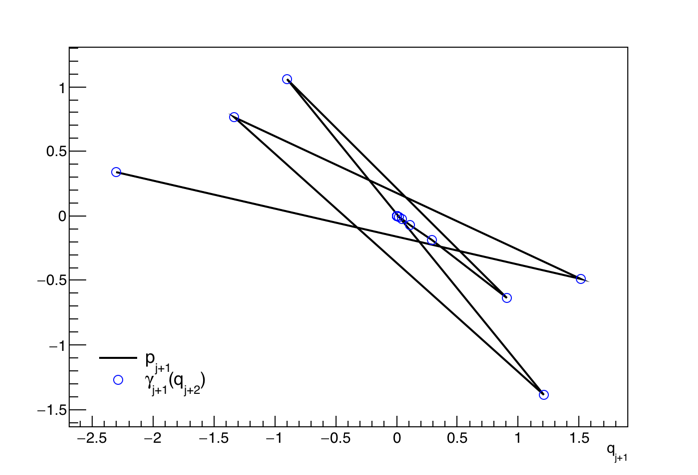}
 \caption{$\gamma_{j+1}(q_{j+2})$ and $p_{j+1}$ vs. $q_{j+1}$}
\end{figure}

\section{Conclusions}

In this paper we have proposed two alternative ways of solving a discrete Hamiltonian problem through two different geometric
interpretations.
The first approach consists of reinterpreting former results available in the literature of discrete Mechanics, by
their geometric understanding based on projected flows and the existence of a generating function whose first-order derivative is
a Lagrangian submanifold of the discrete phase space.
The second approach consists of understanding the discrete dynamics in terms of a discrete vector field whose integral
curves are the discrete Hamilton equations. We propose a geometric interpretation by a projected discrete vector field which
composed with a Lagrangian submanifold of the discrete phase space provides the dynamics of the complete discrete Hamiltonian vector field.
For this matter, we have constructed a discrete Hamiltonian vector field, whose interpretation in the discrete realm is not straightforward.
As a byproduct, we obtain two different discrete Hamilton--Jacobi equations. From the first approach we retrieve the discrete Hamilton--Jacobi
equation existing in the literature. From the second, we obtain a different Hamilton--Jacobi equation which is proven to be equivalent
to the first.
An optimal control example compares the accuracy of the two approaches. It is evident that our interpretation
in terms of discrete vector fields is more accurate than former theories of discrete Mechanics. Evidency is given through
numerical computation and graphic results.
In this way, this manuscript provides an alternative way of obtaining the momenta of a dynamical system through a geometric and
discrete Hamilton--Jacobi theory founded on discrete Hamiltonian vector fields.

\end{example}

\newpage

\section*{Appendix A}

The left discrete action is
\begin{equation}\label{leftHam}
 S_d^{j+1}(q_{j+1})-S_d^{j}(q_j)=-p_jq_j-H_d^{-}(p_j,q_{j+1}).
\end{equation}
If we derivate with respect to $q_j$, we have that $p_j=DS_d^j(q_j)$. Introducing this in the expression,
we obtain the {\it left discrete Hamilton--Jacobi equation}.
\begin{equation}\label{leftdishje}
 S_d^{j+1}(q_{j+1})-S_d^{j}(q_j)+DS_d^{j}(q_j)q_j+H_d^{-}(DS_d^{j}(q_j),q_{j+1})=0
\end{equation}
The left discrete Hamilton--Jacobi equation is equivalent to the right discrete Hamilton--Jacobi equation \eqref{rdhe}.
Their equivalence gives us the relationship between the right and left Hamiltonians,
\begin{equation}\label{eqv}
 H_d^{-}(p_j,q_{j+1})+p_jq_j=H_d^{+}(q_j,p_{j+1})-p_{j+1}q_{j+1}
\end{equation}
The discrete flow interpretation can be reenacted for the left formalism.

\section*{Appendix B}

The left discrete vector field is constructed with the left Hamilton equations \eqref{ldhe}.
In this way,
\begin{equation}
 X_d^{-}=\sum_{i=1}^{N-1}\left(-D_2H_d^{-}(q_{j+1},p_j)\frac{\partial}{\partial q_j}-D_1H_d^{-}(q_{j+1},p_j)\frac{\partial}{\partial p_{j+1}}\right)
\end{equation}
Equivalently, the Hamilton--Jacobi theory can be interpreted through the left vector field as performed in \eqref{comp} for the right case.
The projected vector field is 
\begin{equation}
 X_d^{-}=-\sum_{j=1}^{N-1}D_2H_d^{-}(q_{j+1},p_j)\frac{\partial}{\partial q_j}
\end{equation}
Choosing a section $\gamma=\{\gamma_j(q_{j+1}), j=1,\dots,N-1\}$ and imposing \eqref{comp}, we arrive at

\begin{equation}
 D_2H_d^{-}(q_{j+1},p_j)\frac{\partial \gamma_{j+1}}{\partial q_j}=D_1H_d^{-}(q_{j+1},p_j)
\end{equation}
\noindent

This Hamilton-Jacobi equation is equivalent to the right discrete Hamilton--Jacobi equation in \eqref{midhje}. Furthermore,
this equation can also be obtained through the left discrete flow interpretation in terms of generating functions.

\section*{Acknowledgements}
This work has been partially supported by MINECO MTM 2013-42-870-P and
the ICMAT Severo Ochoa project SEV-2011-0087.

\end{document}